\PassOptionsToPackage{table}{xcolor}
\documentclass[sigconf, nonacm]{acmart}
\usepackage{algorithm}
\usepackage{enumitem}
\usepackage{xcolor} % [table] option passed via \PassOptionsToPackage above (avoids clash with newer acmart that preloads xcolor)

\newcommand{\del}[1]{}
\newcommand{\rtag}[1]{}
\usepackage[normalem]{ulem}
\usepackage{algpseudocode}
\usepackage{bbm}
\usepackage{booktabs}
\usepackage{multirow}
\usepackage{pifont}    %xx
\AtBeginDocument{%
  }

\setcopyright{none}
\copyrightyear{2026}
\acmYear{2026}

\renewcommand\footnotetextcopyrightpermission[1]{}

\begin{document}

\title{Reliable Associative Lookup in Content-Addressable Memory}

\author{Fan Li}
\affiliation{%
  \institution{University of Central Florida}
  \city{Orlando}
  \state{FL}
  \country{USA}
}
\email{fan.li@ucf.edu}

\author{Yanan Guo}
\affiliation{%
  \institution{University of Rochester}
  \city{Rochester}
  \state{NY}
  \country{USA}
}
\email{yanan.guo@rochester.edu}

\author{Xin Xin}
\affiliation{%
  \institution{University of Central Florida}
  \city{Orlando}
  \state{FL}
  \country{USA}
}
\email{xin.xin@ucf.edu}

\begin{abstract}

Content Addressable Memory (CAM) is an important memory paradigm, which performs fast search by comparing an input query against all stored entries in parallel, achieving $O(1)$ lookup complexity. CAM is typically built upon conventional memory technologies, such as SRAM and Non-Volatile Memory (NVM). Accordingly, CAM can also be subject to the reliability challenges of these underlying technologies. In traditional memory systems, protection codes play a critical role in ensuring reliability and have been extensively studied. However, protection codes for CAM have remained largely unexplored. This paper takes an initial step toward addressing this longstanding gap by introducing a non-traditional code design.

\end{abstract}

\maketitle

\section{Introduction}
\label{sec: intro}

Modern memory architecture follows two access paradigms: (1) address-based, i.e., traditional memory, which takes an address as input (e.g., SRAM, DRAM, and NVM), and (2) content-based, i.e., Content Addressable Memory (CAM), which accepts data as input and performs fully parallel associative search, enabling deterministic $O(1)$ lookup latency regardless of data structure size.
Beyond its traditional use in network routing~\cite{cam_routing,cam_routing2} and highly associative cache designs~\cite{cam_cache_1, cam_cache_2, cam_cache_3}, CAM surges as a compelling foundation for modern architectural innovations, including database accelerators~\cite{cam_db_1, cam_db_2, cam_db_4}, sequence alignment engines~\cite{cam_seq_aligement, cam_seq_aligement_2}, and processing-in-memory (PIM) substrates~\cite{CAPE, cam_db_3}.

Despite the architectural differences, CAM is typically constructed on top of conventional memory technologies. For example, CAM can build upon SRAM~\cite{SRAM-CAM, SRAM-CAM2} or NVM~\cite{nvm_cam, cam_nvm_2} cells as its storage core, while augmenting them with additional comparison logic. Consequently, CAM extends the functionality of traditional memory, but also inherits its fundamental limitations. Among these, reliability has been a longstanding challenge, and extensive efforts have been devoted to improving the reliability of memory technologies, including SRAM~\cite{seshadri2014dirty, arm_parity_ecc_caches, wilkerson2010reducing, paul2010reliability}, DRAM~\cite{bamboo, monte_carlo, udipi2012lot, jian2014ecc}, and NVM~\cite{zhang2018exploring, schechter2010use}. As such, it is unlikely that CAM, built upon these technologies, can achieve perfect reliability.

However, relatively few studies have focused on protection codes for CAM, leaving this area largely unexplored. One possible reason is the fundamental difference in failure mechanisms.
In conventional memory, data is explicitly read out, allowing bit errors to be corrected via ECC, i.e., a read-then-check process.
In contrast, CAM does not expose stored data. Errors manifest implicitly through comparison outcomes. This leads to two distinct failure modes: (1) false positives, where non-matching entries are incorrectly reported as matches, and (2) false negatives, where valid matches are missed. False positives are relatively easier to handle, since one can still follow a read-then-check process by reading out the matched entries and performing ECC verification. False negatives, however, are more challenging to address, because missed matches do not appear in the comparison results and thus provide no explicit indication of error.
Consequently, existing protection codes, designed for traditional memory, are fundamentally incompatible with CAM. 

To tackle the problem, a common approach is to employ approximate CAM that supports inexact matching~\cite{harary2024occam, pagiamtzis2006soft, approx_reliable_CAM}. For example, for an $n$-bit single error correction (SEC) codeword, a query with $n-1$ matching bits can be safely treated as a match, as the code guarantees a unique correction for a single-bit mismatch. However, this imposes stringent sensing requirements due to the reduced margin (scaling with $1/n$), leading to significant hardware modification. In contrast, a key advantage of protection codes in conventional memory is that they enhance reliability without fundamentally altering the underlying memory structure.
Alternatively, one can perform $n$ searches, each excluding a different bit position of the $n$-bit query (through CAM's masking operation), thereby enumerating all possible positions of a single-bit error. We regard this $n$-search approach as a baseline. It is functionally equivalent to the approximate CAM solution, which trades sensing complexity for additional searches. However, this degrades the lookup complexity from $O(1)$ to $O(N)$.

To bridge this gap, we propose encoding data using a balanced representation, where each $N$-bit (even $N$) codeword contains an equal number of 1s and 0s (its coding space is $\binom{N}{N/2}$). This property allows each codeword to be uniquely identified by either the positions of 1s or 0s. For example, to search for the codeword `$00110101$', one can construct the query using only the 1s as `$\text{-}\text{-}11\text{-}1\text{-}1$' or using only the 0s as `$00\text{-}\text{-}0\text{-}0\text{-}$', where `$\text{-}$' denotes the bit excluded from comparison. 
This enables tolerance to single-bit errors, as an error can only disrupt one polarity. For instance, a `$1 \rightarrow 0$' bit flips causes a mismatch with `$\text{-}\text{-}11\text{-}1\text{-}1$', but it does not impact the `$00\text{-}\text{-}0\text{-}0\text{-}$' search.
As a result, this balanced-code-based approach, termed the balanced code (BC, Section~\ref{subsec: BC}), ensures single-bit error detection with only two searches ($O$(1) complexity) and avoids altering the underlying memory structure.
Mapping a value into, and back out of, this balanced form (the encoder/decoder) is a standard step that sits \emph{off} the lookup's critical path.

To address multi-bit errors, we extend balanced codes to detect double-bit errors, which leads to the Extended Balanced Code (EBC, Section~\ref{subsec: EBC}).

More broadly, protection-code design for CAM is a systematic research problem. In contrast, ECC for conventional memory has evolved over decades~\cite{mittal2018survey} and it is therefore difficult to comprehensively address all aspects of a CAM code in a single paper. To maintain focus, we emphasize the following primary contribution:

\begin{enumerate}[label=$\bullet$]%[leftmargin=*, nosep]
    \item \,\textbf{The first protection code for associative, match-based storage.} Prior protection codes all guard the read-out path; we identify CAM's dominant \emph{silent} false-negative vulnerability on the match path and introduce the first balanced-code protection scheme designed for content-addressable memory.
\end{enumerate}

\section{Background}
\label{sec:background}

\subsection{CAM Basics}
\label{subsec: cam-basics}

CAM enables fully parallel comparison between an input query and all stored entries. As shown in Figure~\ref{fig: background}(a), the input is broadcast through search lines ($SL_i$ and $\overline{SL}_i$), while the results are returned via match lines ($ML_i$). Note that a lookup may produce zero, one, or multiple matches. In many applications, the final output is not the match signal itself but the associated data readout (e.g., in caches, a tag match triggers the retrieval of the corresponding cacheline).

As a specialized memory structure, CAM is built upon existing memory technologies. Prior work has explored CAM designs based on SRAM~\cite{SRAM-CAM}, NVM \cite{nvm_cam, cam_nvm_2}, and DRAM \cite{dram_cam}, by integrating additional logic. Consequently, CAM typically incurs higher area overhead compared to conventional memory arrays. For example, an SRAM-based CAM can be $2\sim5\times$ larger than a standard SRAM~\cite{jeloka201628}.
Figure~\ref{fig: background}(b) illustrates an SRAM-based CAM cell. Additional peripheral transistors are integrated with the storage core (Figure~\ref{fig: background}(c)) to enable an XNOR operation between the stored bit and the input (i.e., $BL \cdot SL+\overline{BL} \cdot \overline{SL}$). A mismatch causes the match line (ML) to discharge. The ML is connected to a sense amplifier that aggregates the XNOR results across all bits of an entry, and any bit mismatches result in a 0 output value.
Similarly, CAM cells can be constructed using NVM technologies. For example, Figure~\ref{fig: background}(d) shows a design based on two 1T1R NVM cells. 
Note that CAM must retain standard memory access paths (which are not explicitly shown in Figure~\ref{fig: background}(a)), e.g., wordlines ($WL$s) and bitlines ($BL$s) in Figure~\ref{fig: background}(b). Otherwise, CAM entries cannot be updated.

An often \textbf{overlooked aspect} is that the inputs $SL_i$ and $\overline{SL}_i$ need not always be complementary. When both are set to 0, the corresponding bit is effectively disabled, referred to as `masked', allowing it to be ignored during comparison\footnote{Although, in principle, a single-ended input (only $SL_i$) could suffice, this would require adding an inverter within each cell, increasing cell complexity. Therefore, dual-line inputs ($SL_i$ and $\overline{SL}_i$) is common practice~\cite{SRAM-CAM, SRAM-CAM2, nvm_cam, cam_nvm_2}.}.
Building on CAM, ternary CAM (TCAM) further extends functionality by introducing a wildcard state, enabling more flexible matching patterns. For brevity, we focus on CAM in this paper.

\begin{figure}[htbp]
    \centering
    \includegraphics[width=0.45\textwidth]{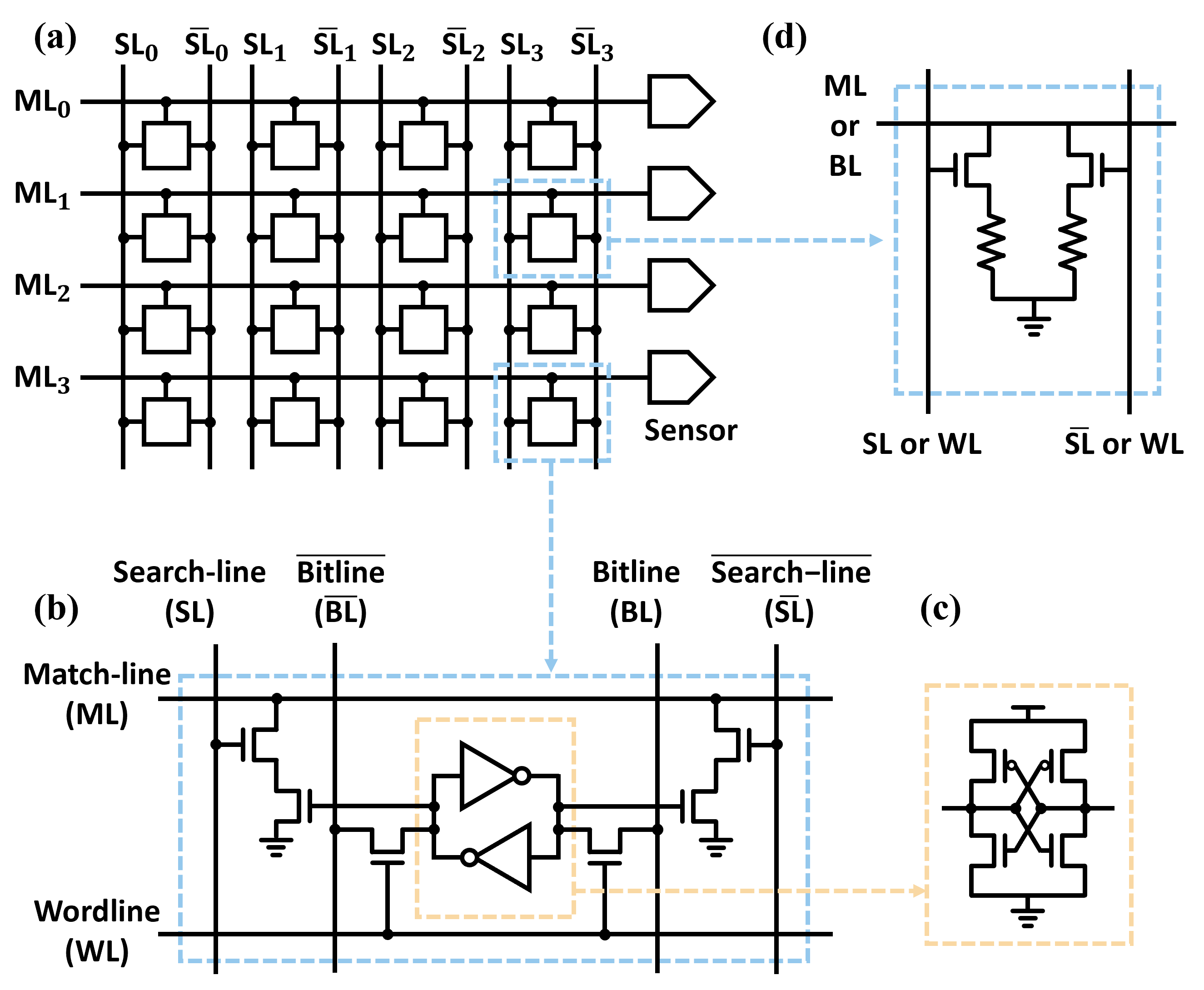}
    \caption{CAM organization: (a) the structure of a $4\times4$ CAM array, (b) a SRAM-based CAM cell, (c) the cross-coupled inverter in an SRAM cell, (d) an NVM-based CAM cell.}
    \label{fig: background}
\end{figure}

\subsection{Traditional Memory Protection}

\subsubsection{General Error Protection} 
\label{subsec: general-error-pro}
Single Error Correction -- Double Error Detection (SEC-DED) is one of the most widely adopted ECC schemes in memory systems~\cite{dell1997white}. It is based on the Hamming code theory, and typically uses 8 parity bits for a 64-bit word to form a 72-bit codeword with a 12.5\% storage overhead. To handle multi-bit errors, Bose-Chaudhuri-Hocquenghem (BCH) codes were introduced, albeit at the cost of higher redundancy~\cite{chien1964cyclic}. For example, BCH(78,64) (a double error correction (DEC) code) protects 64 data bits using 14 parity bits, allowing up to two-bit-error correction but incurring a 21.9\% storage overhead. 

The fundamental principle behind these codes is Hamming distance. To correct $t$ bit errors, a minimum Hamming distance of $d_{\min}=2t+1$ is required. More importantly, error correction capability naturally implies error detection capability. When a code is used for detection \textbf{only}, it can typically identify a larger number of errors than it can correct. In particular, a code that corrects up to $t$ bit errors can detect up to $2t$ bit errors. For example, a Single Error Correction (SEC) code can be repurposed as a Double Error Detection (DED) code, and an SEC-DED code can be used for Triple Error Detection (TED)\footnote{Correction and maximum detection capabilities are mutually exclusive. For example, a code can operate as either SEC–DED or TED, but not both, since the decoder cannot determine the error count before applying correction.}.

To protect burst-type errors (e.g., multiple adjacent bit flips), symbol-level codes have been developed by grouping multiple bits into a single unit, referred to as a \textit{symbol}. Representative schemes include cyclic redundancy check (CRC) for detection and Reed-Solomon (RS) codes for correction~\cite{reed1960polynomial}. Symbol-level and the above bit-level protection target different error characteristics and cannot replace each other. In this paper, we focus on bit-level protection. 

\subsubsection{Uni-directional Error Protection}

Many memory technologies are prone to errors with a strong unidirectional pattern (either $0 \rightarrow 1$ or $1 \rightarrow 0$)~\cite{wen2013cd, wang2018content, aliagha2019react}.
A well-known technique for unidirectional error protection is Berger coding, which counts the number of 1s (or equivalently the number of 0s) in a data word and stores the counting value (more accurately, its bitwise complement) alongside the data. Under unidirectional errors, all bit flips occur in the same direction. By recomputing the number of 1s at the receiver and comparing it with the encoded value, any discrepancy can be reliably detected. 

\subsubsection{Hardware-Centric Resilient Design}

Besides code-based protection, there also exist hardware-centric design strategies for system resilience improvement. For example, a straightforward approach to improving hardware resilience is to deploy three identical units to perform the same operations and compare their outputs for agreement. This is known as Triple-Modular Redundancy (TMR)~\cite{koren2020fault}. 
Another class of techniques is motivated by the Bathtub Curve model~\cite{bathhub}, targeting the ``infant mortality'' phase caused by manufacturing defects. Burn-in testing filters defective units by applying elevated stress (e.g., high temperature or voltage) to induce early failures. Additionally, redundant spares provide a fail-safe mechanism, maintaining functionality even if components fail prematurely in the field.

\section{Motivation}
\label{sec:motivation}

\subsection{Challenges of CAM Protection}

\noindent A CAM cell is built on a conventional storage technology and consequently inherits that technology's failure modes. Because the $di/dt$ of its comparison operation narrows the cell's voltage margin, a CAM is even \emph{more} soft-error-prone than the substrate it is built on~\cite{syafalni2016multiple}. The failure \emph{character}, however, is set by that substrate. \emph{SRAM} cells suffer particle-induced upsets that grow increasingly \emph{multi-cell} at scaled nodes: over $90\%$ of upset events at the $5$\,nm node flip many adjacent cells in \emph{both} directions~\cite{pieper2023study}, and BTI/HCI/RTN aging compounds this with a quasi-permanent ${\sim}100$\,mV threshold shift over a 10-year lifetime~\cite{zhang2019impact} that erodes read/write margins; these flips are thus largely \emph{bidirectional}. In contrast, \emph{NVM} cells such as ReRAM and PCM endure only $10^{6}{\sim}10^{9}$ writes before failing as \emph{permanent} stuck-at faults, and, unlike charge-based memory, are immune to particle-induced soft errors~\cite{schechter2010use, yavits2020wolfram}; their raw bit-error rate reaches $7{\times}10^{-5}{\sim}10^{-3}$~\cite{zhang2018exploring}, and these wear-out faults are strongly \emph{unidirectional}. \emph{STT-MRAM} instead offers near-unlimited write endurance above $10^{15}$ cycles~\cite{kan2017study}, so its dominant read-disturb, write-error, and retention faults are \emph{transient} and a rewrite corrects them~\cite{aliagha2019react}. \emph{DRAM} exhibits transient bit flips and asymmetric, \emph{unidirectional} retention failures~\cite{wang2018content, wen2013cd}; large field studies find errors dominated by hard, repeat-address faults and ${\sim}79\%$ single-bit, with the multi-bit, symbol, and chip minority delegated to in-memory ECC~\cite{meza2015revisiting, sridharan2015memory}. Across these technologies, every fault ultimately manifests as a bit flip characterized by just two parameters, its \emph{count} and its \emph{direction}. These are precisely the parameters on which our scheme's coverage depends, and they let us map each physical fault class to its bit-level equivalent.

More importantly, in stark contrast to traditional memory architectures, which explicitly read out stored data using addresses, CAM operates by comparing input data against all stored entries in parallel and reporting matches.
This fundamental difference introduces unique protection challenges. For example, as shown in Figure~\ref{fig: motivation}(a), a CAM array returns two matches. However, errors may alter the results, leading to two types of failures: (1) a non-matching entry is incorrectly reported as a match (the false positive case in Figure~\ref{fig: motivation}(b)), and (2) a matching entry is incorrectly recognized as a mismatch (the false negative case in Figure~\ref{fig: motivation}(b)). 

\begin{figure}[ht]
    \centering
    \includegraphics[width=0.5\textwidth]{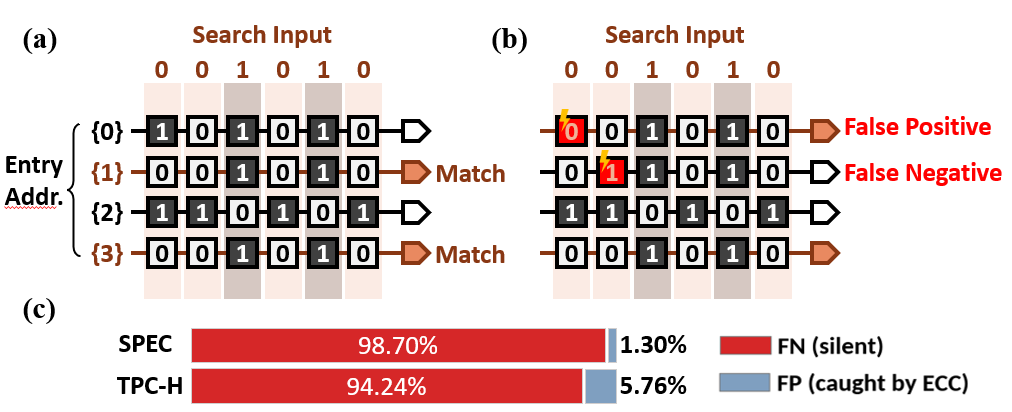}
    \caption{A CAM search operation under (a) error-free and (b) erroneous scenarios; (c) false-negative vs.\ false-positive composition of CAM error events under simulation; the two bars are the geometric mean over typical SPEC and TPC-H workloads, respectively (configuration in Section~\ref{sec:exp_method}). The silent false negative dominates in both.}
    \Description{}
    \label{fig: motivation}
\end{figure}

To address the first challenge, ECC parity (e.g., SEC) can be stored alongside each entry. Upon a match, correctness can be verified by retrieving the associated parity and checking whether it is consistent with the input data. However, the second challenge, false negatives, is both far more likely \emph{and} harder to handle. It is more likely because \emph{any} error in a queried entry deterministically produces a false negative, whereas a false positive requires a non-matching entry to be corrupted into \emph{exactly} the queried key, one specific value out of $2^W$, which is vanishingly improbable (and, in the rare case it occurs, is caught by the read-then-check above). It is harder because, unlike a false positive, a false negative is silent, leaving no direct indication of failure. In this context, traditional ECC schemes are not directly applicable~\cite{syafalni2016multiple}. We confirm this empirically via simulation across two CAM use cases: in the geometric mean over typical SPEC and TPC-H workloads, false negatives constitute $94{\sim}99\%$ of the CAM error events (Figure~\ref{fig: motivation}(c)). This confirms that the silent false negative, not the false positive, is the dominant failure mode.

In addition, false negatives can be particularly dangerous in certain scenarios. For example, in a write-buffer structure that uses CAM to detect pending writes (or in a fully associative cache that uses CAM to locate dirty cache lines), a false negative may cause a valid match to be missed, leading the system to incorrectly assume that no prior write exists. Consequently, a subsequent read may access stale data from lower memory (i.e., a read-after-write (RAW) violation). Such errors can propagate without immediate detection and eventually lead to system failures.

\subsection{Approximate CAM}

Approximate CAM was proposed for applications that require inexact matching, such as sequence alignment, where queries may differ slightly from stored entries~\cite{crafton202428nm, zhong2023asmcap, garzon2023approximate}. Unlike conventional CAM, it allows a bounded number of bit mismatches (e.g., within a given Hamming distance).
Prior work proposed to leverage this customized CAM to tackle the false negative challenge~\cite{harary2024occam, pagiamtzis2006soft, approx_reliable_CAM}. It first uses approximate matching to retrieve `near-matching' candidates, and then applies ECC to verify correctness and filter out incorrect matches.

However, enabling approximate matching introduces significant hardware overhead. Conventional matching only requires resolving a binary outcome (match or mismatch). In contrast, approximate matching must distinguish small differences corresponding to varying numbers of mismatches. For example, tolerating a single-bit difference requires resolving at least three levels (i.e., exact match, one-bit mismatch, and larger mismatches). This necessitates a high-precision ADC per row. For a CAM with $n$-bit entries, the ADC must resolve differences on the order of $1/n$, which is prohibitively expensive. For example, with 64-bit entries, the ADC overhead alone can exceed 200\% of the baseline CAM array (see Section~\ref{sec:exp_area_energy}), which is even less efficient than the TMR solution.
More importantly, introducing ADC further complicates protection, as ADC itself is sensitive to deployment conditions (e.g., IR drop) and environmental variations (e.g., temperature). To ensure that the ADC does not become the bottleneck in a resilient CAM design, even higher precision (e.g., $1/2n$) may be required, exacerbating the already significant hardware overhead.

\subsection{Design Goal}

Protection codes have been widely adopted in memory systems because they provide strong and flexible error coverage with relatively low overhead~\cite{lin2001error}. In particular, they can be adapted to different reliability requirements through various coding schemes, accommodating diverse deployment environments and application needs. In contrast, physically enhanced designs (e.g., approximate CAM structures) often incur significant cost overheads and lack adaptability to evolving error mechanisms.

However, CAM is fundamentally incompatible with conventional memory-oriented ECC schemes, and currently lacks an effective protection mechanism. The \textbf{goal} of this paper is to develop a CAM-specific protection code and use it to reconcile the fundamental misalignment between CAM and traditional memory-oriented ECC.

\noindent Beyond detection difficulty, reliability is itself a barrier to CAM \emph{scaling}. Conventional memory has ridden decades of protection-code advances (e.g., on-die ECC that improves yield and enables continued DRAM scaling~\cite{jian2014ecc, bamboo}), whereas CAM, lacking any comparable protection code, is kept small and special-purpose to bound its error exposure, on top of its already steep area and power cost ($2{\sim}5\times$ SRAM~\cite{jeloka201628}). Reliability, area, and power jointly cap how large a CAM can grow. Meanwhile, CAM's role keeps expanding into the new domains noted in Section~\ref{sec: intro}. Our balanced-code protection delivers reliability at low area and power overhead, and consequently relieves the reliability constraint within these tight budgets. This fills a long-standing gap and lets CAM scale.

\begin{figure}[ht]
    \centering
    \includegraphics[width=0.425\textwidth]{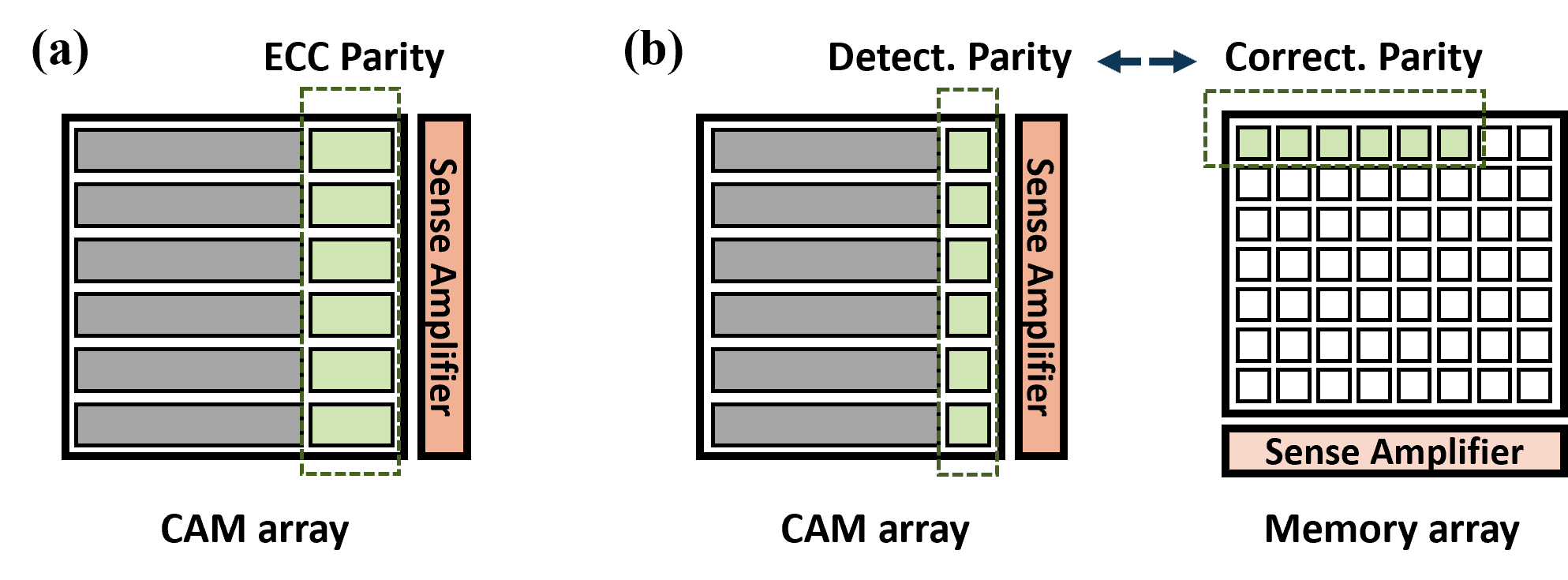}
    \vspace{-0.1in}
    \caption{CAM ECC storage strategies: (a) inline ECC and (b) decoupled ECC storage.}
    \Description{}
    \label{fig: decouple}
\end{figure}

\section{Balanced-Code Design}
\label{sec:method}

Before introducing our main design, we first highlight a common practice in traditional memory ECC design: decoupling error detection from correction~\cite{yoon2010virtualized, udipi2012lot, nair2016xed, aspa, saileshwar2018synergy}. Since errors are relatively rare per access (rates are technology-dependent; see the fault model in Section~\ref{sec:motivation}), most operations only require detection, while correction is invoked only when errors are detected. This reduces decoding latency and parity transmission overhead.
We adopt the same principle. As shown in Figure~\ref{fig: decouple}(b), parity in CAM is used solely for detection, while the corresponding correction parity is stored in memory. This reduces the overall cost of parity storage, as memory is more cost-effective than CAM (Section~\ref{subsec: cam-basics}). It also enables optimizing correction parity independently. For example, multiple entries can be combined into a larger codeword to improve the correction capability, similar to tiered ECC strategies~\cite{zhang2018exploring, udipi2012lot, bamboo, jian2013low}.

However, such decoupling-related optimizations have been extensively studied and are not the focus of this work. We adopt this decoupling primarily to simplify the presentation of our proposed design. In other words, our design also supports direct ECC storage in CAM (Figure~\ref{fig: decouple}(a)).
\textbf{In the following, we focus on the error detection process.} Correction is performed by retrieving the associated ECC parity from memory. The correction capability is aligned with the CAM detection capability, e.g., in-CAM SED (single-error-detection) with in-memory SEC (single-error-correction), and in-CAM DED (double-error-detection) with in-memory DEC (double-error-correction).
Next, we begin with a straightforward design that performs error detection for CAM through a sequence of searches, each examining a single bit position. We then present our core design, which exploits a non-traditional code, termed balanced code, to reduce the process to only two searches.

\begin{figure}[ht]
    \centering
    \includegraphics[width=0.48\textwidth]{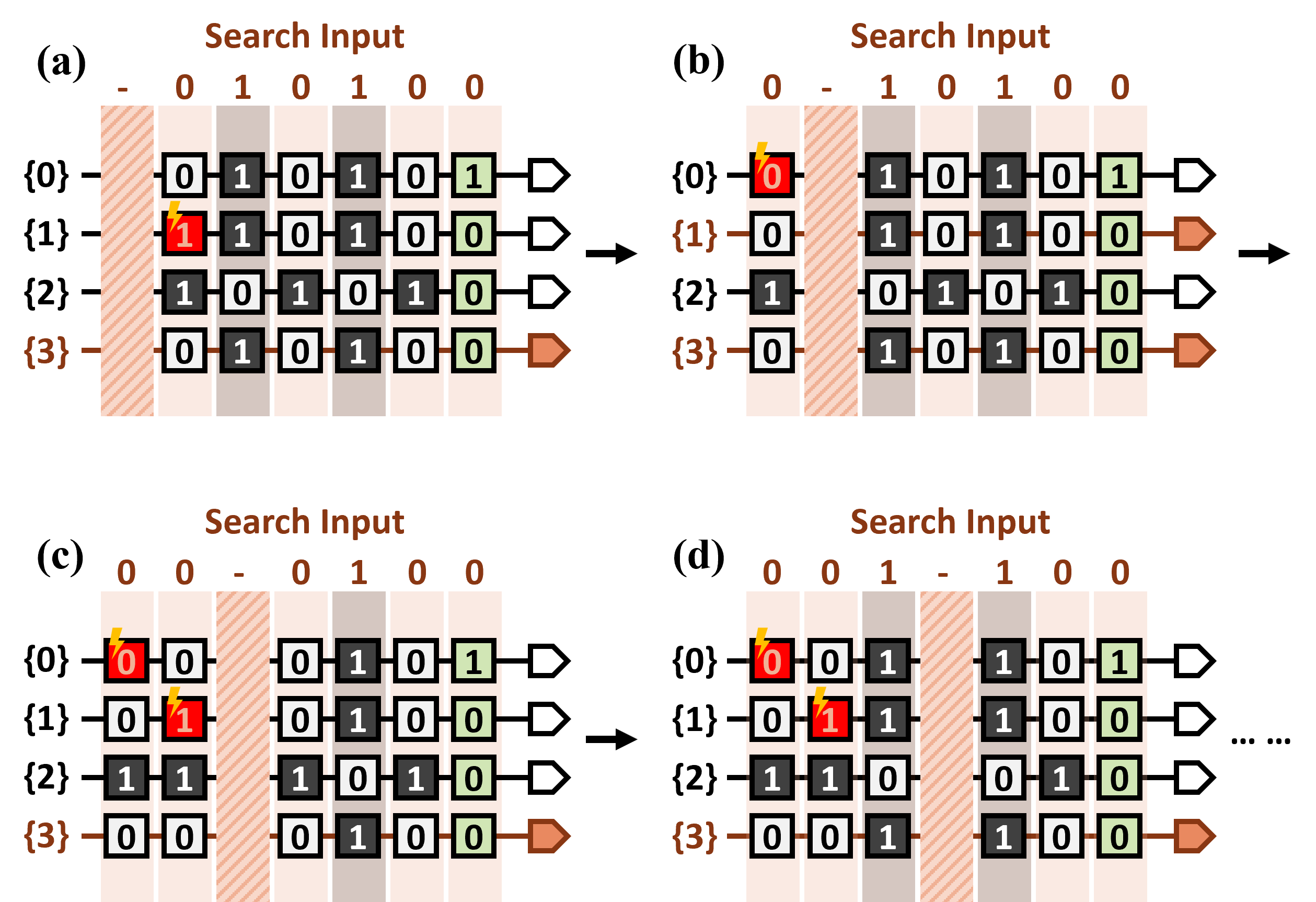}
    \caption{The mechanism of the brute-force masked search (Baseline): masking one bit of the search input and enumerating all bit positions.}
    \Description{}
    \label{fig: Baseline}
\end{figure}

\subsection{Baseline: Simple CAM-Friendly Design}
\label{subsec: Baseline}

As introduced in Section~\ref{subsec: cam-basics}, each CAM input bit can take three states: 1, 0, or a disabled (referred to as {masking}) state. Based on this observation, we introduce a simple brute-force masked search approach (the \textbf{Baseline}) to detect a single-bit error by identifying its location bit by bit.
Figure~\ref{fig: Baseline} illustrates an example using a single-error-detection (SED) code (i.e., even or odd parity), where each 6-bit CAM entry is extended with one parity bit (e.g., even parity over the 6 bits). In each search, one of the 7 bit positions is masked. If the error resides in the masked position, the remaining bits align with the input, resulting in a match. Let the output of the $i$-th search be denoted as $S_i$ (note that a CAM array may contain duplicate entries, yielding multiple matches). By iterating over all 7 bit positions, we obtain 7 output sets.
The intersection of these sets corresponds to error-free matches, denoted as the clean set $S_c$. The set difference, denoted as $S_e$, captures entries with a single-bit error. The set $S_e$ is then forwarded to the next stage, where the corresponding correction parity is retrieved from memory for error correction. 
It is also worth noting that, in most cases, $S_e$ is empty, as errors rarely occur.
The same process can also support an SEC code. In this case, no suspected set $S_e$ exists, since the masked bit can be uniquely determined by SEC. As noted earlier, we focus on detection and omit further discussion of using correction parity in CAM in the following designs.

This sequential search approach is compatible with conventional error detection (or correction) codes,  but it introduces additional latency. For single-bit error detection, it requires $n$ searches for an $n$-bit codeword, which remains practical for small $n$. However, the more fundamental challenge arises when extending it to multi-bit error detection, e.g., DED and TED (Section~\ref{subsec: general-error-pro}), as the complexity grows combinatorially. For example, detecting two-bit errors requires masking all possible pairs of bit positions, resulting in $\binom{n}{2}$ searches (i.e., $\sim\!n^2$). Similarly, detecting three-bit errors requires $\binom{n}{3}$ searches (i.e., $\sim\!n^3$). This quickly becomes impractical, motivating a more efficient design.

\begin{figure}[ht]
    \centering
    \includegraphics[width=0.45\textwidth]{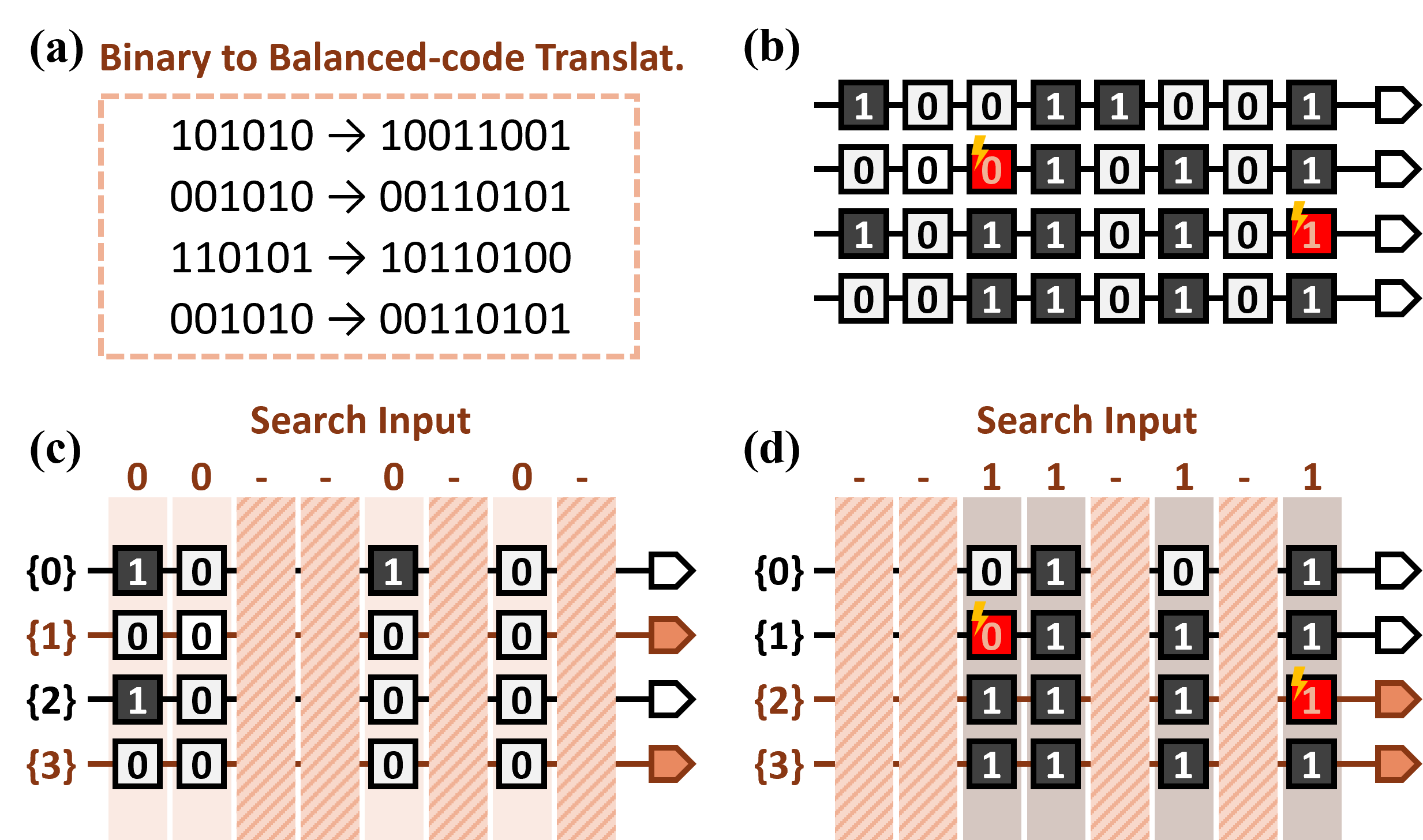}
    \caption{The mechanism of the balanced code (BC): (a) binary numbers and associated balanced code representations, (b) error scenarios, and (c) 0s-only and (d) 1s-only searches.}
    \Description{}
    \label{fig: BC}
\end{figure}

\subsection{Balanced Code (BC) for CAM}
\label{subsec: BC}

To address the above challenge, we propose \textbf{BC}, which represents data with an equal number of 0s and 1s. This format, commonly referred to as a DC-balanced or constant-weight code, is used to maintain stable average voltage or current~\cite{widmer1983dc, immink2010very, DC_wander_effect}. For simplicity, we refer to it as a balanced code in this paper. 

A key property of balanced codes is that they can detect \textbf{any number} of uni-directional errors (e.g., $1 \rightarrow 0$ flips and no simultaneous $0 \rightarrow 1$ flips), as such errors break the equal number of 0s and 1s in the codeword.
Correspondingly, in the absence of errors, identifying either the 1s or the 0s is sufficient to fully determine the codeword. For example, for an $N$-bit balanced codeword (where $N$ is even), once the $N/2$ bits of 1s are known, all remaining $N/2$ bits must be 0s (and vice versa).
In the case of a single-bit flip, it naturally constitutes a one-bit uni-directional error. Thus, balanced codes can be used to detect single-bit errors. Equivalently, balanced codes guarantee a minimum Hamming distance of 2.

The above property aligns well with CAM, as the masking operation enables selective comparison of bit positions. 
Taking the four entries in Figure~\ref{fig: motivation}(a) as an example, each 6-bit word is first encoded into an 8-bit balanced codeword, e.g., `$001010$' $\rightarrow$ `$00110101$', as shown in Figure~\ref{fig: BC}(a). Here, an $N$-bit balanced code provides a coding space of $\binom{N}{N/2}$. Accordingly, an 8-bit balanced code can provide $\binom{8}{4} = 70$ codewords, which are sufficient to represent all $2^6 = 64$ possible 6-bit values.
To search for the codeword `$00110101$', one can set the input either as `$\text{-}\text{-}11\text{-}1\text{-}1$' or `$00\text{-}\text{-}0\text{-}0\text{-}$', where `$\text{-}$' represents a masked bit. Either input uniquely determines the codeword in the absence of errors, while using both inputs provides tolerance to a single-bit error. Specifically, a single-bit flip affects either a `1' or a `0', so one of the two masked searches remains unaffected and still produces a correct match. 
As shown in Figure~\ref{fig: BC}(b), suppose entries {1, 2} each contain a single-bit flip. The detection process is as follows (Figure~\ref{fig: BC}(c) and (d)):

\begin{enumerate}[label=$\bullet$]
    \item Searching  with `$00\text{-}\text{-}0\text{-}0\text{-}$', referred to as 0s-only search, produces a match set of \{1, 3\}. % (i.e., entry\_1, entry\_3). 
    \item Searching with `$\text{-}\text{-}11\text{-}1\text{-}1$', referred to as 1s-only search, produces a match set of \{2, 3\}. % (i.e., entry\_2, entry\_3). %, where codeword\_2 is a false positive. 
\end{enumerate}

Correspondingly, the clean set $S_c$ (the intersection of the above two match sets) is \{3\}, and the suspected set $S_e$ (the difference between the two match sets) is \{1, 2\}. The set $S_e$ is then forwarded to the next error-correction stage.
In addition to single-bit error detection, the above process can naturally \textbf{detect uni-directional errors}. For example, a codeword with any number of $1 \rightarrow 0$ errors will still be included in the 0s-only search results, but excluded from the 1s-only search results, thereby distinguishing it from valid codewords.
Notably, compared to the $n$-search brute-force approach, BC reduces the single-bit error detection process to only two searches.

\section{Balanced-Code Optimization}

BC detects single-bit and uni\-direc\-tional errors in two searches, but extending balanced codes to general multi-bit detection by sequential masked search would again cost $O(N)$. To detect 2-bit errors in constant time, we rethink the role of `detection' in CAM and propose an extended balanced code with an over-inclusive detection scheme. Next, we first use 2-bit error detection as an example to explain the key idea of over-inclusive detection, and then introduce the proposed alternative.

\subsection{Over-inclusive Error Detection}

In conventional coding theory, to detect up to $t$ errors, a code must have minimum Hamming distance $d_{min} \geq t+1$. This strict requirement ensures that any corrupted codeword can be uniquely distinguished from all valid codewords. Otherwise using a code with $d_{min} = t+1$ to detect more than $t$ errors, certain error patterns may become indistinguishable from valid codewords, leading to silent data corruption (SDC).
However, in the context of CAM, the target codeword is already known, and detection is performed through search rather than direct decoding. To tolerate $t$ errors, each masked search input should differ from the target codeword in at least $t+1$ bit positions. If codewords stored in CAM has limited Hamming distance ($d_{min} < t+1$), the masked search may return a valid codeword as a match, wrongly classifying it as an invalid one.

However, this does not lead to an incorrect final result. Recall that all entries identified as invalid (classified in $Sc$) are subsequently retrieved for further error-correction checking. During this additional checking step, such falsely flagged entries can be distinguished from truly corrupted codewords.
On the other hand, enforcing large Hamming distance in balanced codes incurs significant overhead. If the collision rate (i.e., the possibility of wrongly classifying a valid codeword) is sufficiently low, we can relax the Hamming distance requirement ($d_{min} \geq t+1$) to enable more efficient balanced code designs with modest overhead.

Consider detecting errors by performing bit-by-bit masked searches over a balanced code alone (without additional parity). A single-bit error is still exposed, but the search may also flag some valid codewords as suspects, with potential ambiguity. We refer to this process as over-inclusive detection.
For example, suppose `$000011111$' and `$10000111$' are both stored in the CAM. When searching for `$000011111$' using a masked pattern such as `$\text{-}\text{-}\text{-}\text{-}\text{-}111$', both entries will be returned as matches. 
Finally, this `$10000111$' will be classified as an invalid codeword (included in set $S_c$), since it does not match in the other iterations of searches (e.g., `$\text{-}\text{-}\text{-}\text{-}1\text{-}11$'). 
Likewise, if other codewords such as `$01000111$', `$00100111$', and `$00010111$' are stored, they will also be returned as matches. In total, there are $\frac{N}{2}$ such codewords, where $N$ is the codeword length. Given the codeword space $\binom{N}{N/2}$, the collision rate is $\frac{N}{2\binom{N}{N/2}}$, which becomes negligible as $N$ increases.

However, this over-inclusive detection still requires $O(N)$ search complexity. To further reduce this overhead, we propose a more efficient balanced code tailored for this detection scheme.

\subsection{Extended Balanced Code (EBC)}
\label{subsec: EBC}

In contrast to conventional balanced codes (which enforce equal numbers of 1s and 0s), we construct a new balanced code (referred to as extended-balanced code) by enforcing equal counts of four 2-bit patterns: 00, 01, 10, and 11. For ease of presentation, we use letters $a$, $b$, $c$, and $d$ to represent these four patterns, respectively. Figure~\ref{fig: EBC}(a) illustrates an example, where each 2-bit pattern appears twice, resulting in a codeword in length of 16 bits. With $4k$ 2-bit patterns (i.e., $8k$-bit codeword length), the codeword space is $\frac{(4k)!}{(k!)^4}$. For example, when $k=9$ (i.e., $72$-bit codewords), the resulting 72-bit codeword space is $\sim2.15\times10^{19}$, which is sufficient to encode the entire 64-bit binary space ($2^{64}$). 

\begin{figure}[ht]
    \centering
    \includegraphics[width=0.425\textwidth]{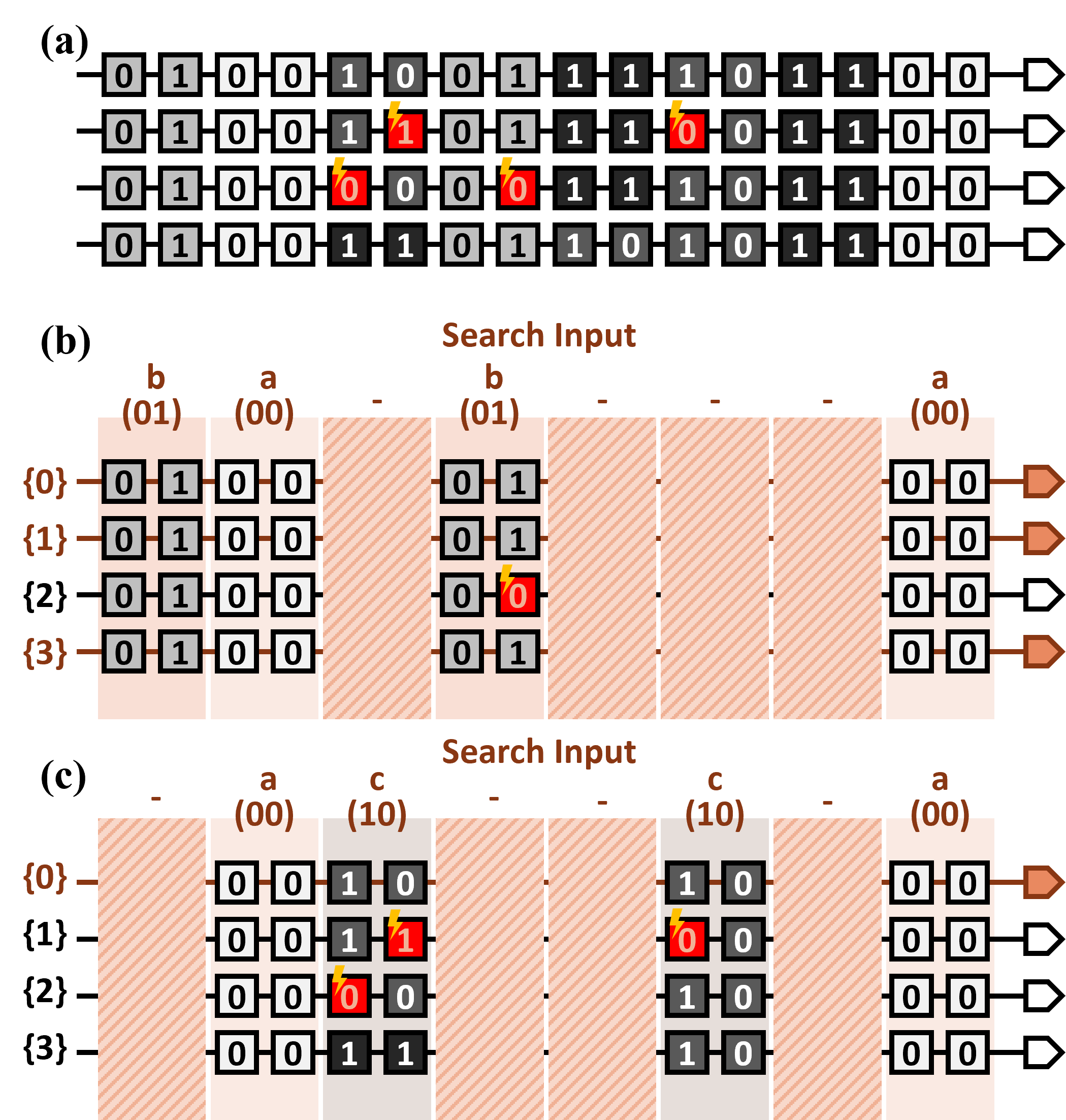}
    \caption{The mechanism of the extended balanced code (EBC): (a) a four-entry CAM array encoded in the proposed extended balanced-code format, where the first three entries store the \emph{same} codeword purely to illustrate different 2-bit error cases, and two entries are corrupted with 2-bit errors; (b) an $ab$-only search and (c) an $ac$-only search. All six pairwise searches probe each single stored codeword in place; they are repeated searches, not replicated copies.}
    \Description{}
    \label{fig: EBC}
\end{figure}

This construction inherently preserves the balance between 1s and 0s, and thus naturally supports single-bit error detection. We next focus on 2-bit error detection. A 2-bit error can affect at most two of the four pattern types. Based on this observation, we propose a protection strategy, termed EBC, which performs six masked searches. Each search preserves two pattern types while masking the other two. Using the example shown in Figure~\ref{fig: EBC}(b) and (c), the detection process is as follows:

\begin{enumerate}[label=$\bullet$]%[leftmargin=*, nosep]
    \item Searching with `$ba\text{-}b\text{-}\text{-}\text{-}a$', referred to as an $ab$-only search (masking $c$ and $d$), produces a match set of \{0, 1, 3\}. 
    \item Searching with `$\text{-}ac\text{-}\text{-}c\text{-}a$', referred to as an $ac$-only search (masking $b$ and $d$), produces a match set of \{0\}.
    \item Applying the same pattern to the remaining combinations (i.e., $ad$-, $bc$-, $bd$-, and $cd$-searches) produces four additional match sets: \{0, 1, 2\}, \{0\}, \{0, 1\}, and \{0\}.
\end{enumerate}

Similarly, the clean (error-free) set is $S_c$=\{0\}, i.e., the intersection of the six match sets. In contrast, the suspected set is $S_e$=\{1, 2, 3\}, i.e., the difference between the six match sets.  %corresponds to suspected entries with up to 2-bit errors. 
Note that the last entry \{3\} is another valid codeword, but it is wrongly classified into $S_e$. This is due to the over-inclusive error detection. Since the proposed extended-balanced code still has a minimum Hamming distance of 2, as discussed earlier, using it for 2-bit error detection may introduce collisions with other valid codewords. Specifically, the last entry \{3\} can be represented as `$badbccda$', which is returned as a match during the $ab$-only search of `$ba\text{-}b\text{-}\text{-}\text{-}a$' (Figure~\ref{fig: EBC}(b)). Theoretically, with $4k$ letters, each search can match up to $\binom{4k}{2k}$ valid codewords. The collision rate of each search is $\binom{2k}{k}/\frac{(4k)!}{(k!)^4}$. 
This rate becomes negligible as $k$ increases. For example, when $k=9$ (i.e., 72-bit codewords), the collision rate is about $2^{-48}$.
Notably, compared to the $O(N)$ over-inclusive (sequential) scheme, EBC reduces error detection to a constant number of only 6 searches.

\begin{theorem}[Constant-time 2-bit detection]
\label{thm:ext}
Let $x$ be an extended-balanced codeword with a fixed composition over the four 2-bit patterns $\{a,b,c,d\}$. Every weight-preserving 2-bit fault in $x$ is detected by the six pairwise masked searches $\{ab, ac, ad,$ $bc, bd, cd\}$.
\end{theorem}
\begin{proof}
Partition $x$ into $S$ two-bit slots, each holding one of the patterns $\{a,b,c,d\}$ with a per-pattern multiplicity (composition) fixed by the code. A weight-preserving 2-bit fault flips one bit $1{\to}0$ and one $0{\to}1$; whether the two flips fall in one slot or two, at least one slot's two-bit symbol changes, so some query pattern $p$ loses a slot and its multiplicity strictly drops. Each pairwise search ``$pq$'' imposes the constraint that every $p$-slot reads $p$ and every $q$-slot reads $q$ (all other slots masked); the changed slot, which the query labels $p$ but now reads $p'{\neq}p$, violates this constraint, so every search retaining $p$ mismatches. Pattern $p$ occurs in three of the six pairs ($pa,pb,\dots$), so at least three searches fail and the entry is excluded from the clean set~$S_c$ and is detected; an uncorrupted codeword satisfies all six constraints. The scheme is over-inclusive (Section~\ref{subsec: EBC}): a distinct valid codeword may also be excluded, with probability $\le\!2^{-48}$ at $k{=}9$, but it is filtered by the subsequent ECC check, so no fault is missed.
\end{proof}

\begin{table}[th]
      \centering
      \caption{Balanced-code design comparison in terms of codeword length (L) and search count (\#) for 16-, 32-, and 64-bit data}
      \label{tab: bc_compare}
      \small
      \setlength{\tabcolsep}{3pt}
      \begin{tabular}{l c r r r r r r r r}
        \hline \hline
        \rowcolor[gray]{0.9}
        & \textbf{ERR}
        & \multicolumn{2}{c}{\textbf{16-bit}}
        & \multicolumn{2}{c}{\textbf{32-bit}}
        & \multicolumn{2}{c}{\textbf{64-bit}}
        &
        \\
        \rowcolor[gray]{0.9}
        \textbf{Design}
        & \textbf{Cov.}
        & \textbf{L} & \textbf{\#}
        & \textbf{L} & \textbf{\#}
        & \textbf{L} & \textbf{\#}
        & \textbf{Code Scheme}
        \\
        \hline
        Brute-Force
          & 1b
          & 17 & 19
          & 33 & 33
          & 65 & 65
          & even parity
          \\
        BC
          & 1b
          & 19 & 2
          & 36 & 2
          & 68 & 2
          & balanced code
          \\
          \hline
        Brute-Force
          & 2b
          & 21 & 210
          & 38 & 703
          & 71 & 2485
          & hamming DED
          \\ 
          EBC
          & 2b
          & 22 & 6
          & 36 & 6
          & 72 & 6
          & extend. balanced
          \\
        \bottomrule
      \end{tabular}
    \end{table}

\subsection{Balanced Code Length Reduction}

In practice, the total bit count of a balanced codeword does not need to be even. For example, to encode 8-bit data, the minimum $N$ satisfying $\binom{N}{\lfloor N/2 \rfloor} \ge 2^8$ is 11. This yields an 11-bit codeword with five 1s and six 0s, i.e., a constant-weight-5 codeword (where the weight denotes the number of 1s). Although the numbers of 1s and 0s are not equal, all valid codewords still have a fixed number of 1s. As a result, any single-bit error changes this weight and can be detected, providing the same error detection capability as `strictly' balanced codes.
Similarly, the extended-balanced code can adopt the same strategy: it is not necessary to enforce equal counts of the four letters ($a, b, c, d$), as long as each codeword maintains a fixed composition of these letters.

Taking a further step, we can relax the constant-weight constraint of balanced codes to reduce overhead. Instead of using only codewords with weight $\lfloor N/2 \rfloor$, we can draw codewords from two nearby weight classes. For example, to encode 8-bit data, we use 10-bit codewords by selecting 128 constant-weight-4 codewords (four 1s and six 0s) from the $\binom{10}{4}$ space and another 128 constant-weight-6 codewords from the $\binom{10}{6}$ space. This construction preserves a minimum Hamming distance of 2 among all valid codewords, while relaxing the uniform constant-weight constraint. As a result, single-bit error detection remains intact. However, it compromises the capability to fully direct uni-directional errors. For example, both `$0001111000$' and `$0011111100$' are valid codewords under this construction, where the former can be transformed into the latter through two $0 \rightarrow 1$ flips. For simplicity, we present this balanced code, constructed from two constant-weight classes, as an optional optimization and do not incorporate it into the main design.

Table~\ref{tab: bc_compare} summarizes all our balanced-code schemes, along with their corresponding codeword lengths and search counts for error detection. Using the 64-bit data protection as an example and the brute-force masked search as the \textbf{baseline}, for single-bit error detection, BC slightly increases the codeword length from 65 to 68, while reducing the search count from 65 to a constant 2. For 2-bit error detection, EBC uses a 72-bit codeword while reducing the search count from 2485 to a constant 6.

\section{Evaluation}
\label{sec:eval}

\subsection{Methodology}
\label{sec:exp_method}

In this section, we evaluate the hardware cost of our scheme.
We consider the following schemes:
\begin{itemize}[leftmargin=*,align=left]
    \item \textbf{Ideal (No Protection)}: the unprotected CAM, used as the reference, i.e., the area upper bound with no error detection.
    \item \textbf{Baseline (Brute-Force)}: a simple brute-force masked search (Section~\ref{subsec: Baseline}), which is theoretically compatible with ECC schemes without hardware modification. It provides 1-bit error coverage with in-CAM SED parity and in-memory SEC parity, and 2-bit coverage with in-CAM DED parity and in-memory DEC parity.
    \item \textbf{Ours}: (1) BC, which provides 1-bit error coverage with in-CAM balanced code and in-memory SEC parity, and (2) EBC, which provides 2-bit error coverage with in-CAM extended balanced code and in-memory DEC parity.
    \item \textbf{Approximate CAM}: prior ADC-based designs (OCCAM~\cite{harary2024occam}, SpaceCAM~\cite{approx_reliable_CAM}) that tolerate errors within a single search but require a per-row ADC; we model both ADC-SAR and ADC-Flash implementations. With SEC/DEC ECC they provide 1-/2-bit coverage.
\end{itemize}

\paragraph{\textbf{Hardware Overhead.}}
Hardware overhead here includes both storage overhead and hardware-modification-induced overhead. As aforementioned, memory is more cost-effective than CAM. It is more efficient to separate the storage of detection parity (e.g., SED, DED) and correction parity (e.g., SEC, DEC) into CAM and memory, respectively. All schemes follow this strategy, and memory parity overhead is normalized to equivalent CAM overhead.

In terms of hardware modification, approximate CAM incurs additional overhead due to modifications required for inexact matching. To evaluate the associated ADC overhead, we use NeuroSIM V1.4~\cite{neurosim}, considering both SAR~\cite{SAR_ADC} and flash ADCs~\cite{flash_ADC} as representative implementations.
We use the 32\,nm technology node as the reference point. We choose three representative CAM technologies: SRAM-based CAM, ReRAM-based CAM, and STT-MRAM-based CAM. The latter two are also known as NVM-based CAM~\cite{cam_nvm_1, cam_nvm_2}. We use CACTI 7.0~\cite{cacti} and NVSim~\cite{nvsim} for memory modeling and deriving the baseline area.

\subsection{Results}
\label{sec:exp_results}
\label{sec:exp_area_energy}

\begin{table}[th]
  \centering
  \caption{Hardware overhead comparison (SRAM-based)}
  \label{tab:area_breakdown}
  \small
  \setlength{\tabcolsep}{3pt}
  \begin{tabular}{l c r r r r r r}
    \hline \hline
    \rowcolor[gray]{0.9}
    & \textbf{ERR}
    & \multicolumn{2}{c}{\textbf{CAM-only}}
    & \multicolumn{2}{c}{\textbf{CAM+Mem}}
    & \multicolumn{2}{c}{\textbf{Overall}}
    \\
    \rowcolor[gray]{0.9}
    \textbf{Scheme}
    & \textbf{Cov.}
    & \textbf{$\mu$m$^2$} & \textbf{Ovh.}
    & \textbf{$\mu$m$^2$} & \textbf{Ovh.}
    & \textbf{$\mu$m$^2$} & \textbf{Ovh.}
    \\
    \hline
    Ideal (No Prot.)
      & ---
      & 39{,}105 & ---
      & 39{,}105 & ---
      & 39{,}105 & ---
      \\
    \hline
    Brute-Force
      & 1b
      & 39{,}687 & +1.5\%
      & 41{,}227 & +5.4\%
      & 41{,}227 & +5.4\%
      \\
    Brute-Force
      & 2b
      & 43{,}178 & +10.4\%
      & 46{,}177 & +18.1\%
      & 46{,}177 & +18.1\%
      \\
    BC
      & 1b
      & 41{,}432 & +6.0\%
      & 42{,}973 & +9.9\%
      & 42{,}973 & +9.9\%
      \\
    EBC
      & 2b
      & 43{,}759 & +11.9\%
      & 46{,}759 & +19.6\%
      & 46{,}759 & +19.6\%
      \\
    \hline
    ADC-SAR
      & 1b
      & 39{,}687 & +1.5\%
      & 41{,}227 & +5.4\%
      & 350{,}941 & +797\%
      \\
    ADC-SAR
      & 2b
      & 43{,}178 & +10.4\%
      & 46{,}177 & +18.1\%
      & 355{,}891 & +810\%
      \\
    ADC-Flash
      & 1b
      & 39{,}687 & +1.5\%
      & 41{,}227 & +5.4\%
      & 167{,}746 & +329\%
      \\
    ADC-Flash
      & 2b
      & 43{,}178 & +10.4\%
      & 46{,}177 & +18.1\%
      & 172{,}695 & +342\%
      \\
    \bottomrule
  \end{tabular}
\end{table}
\begin{figure}[!htb]
  \centering
  \includegraphics[width=\linewidth]{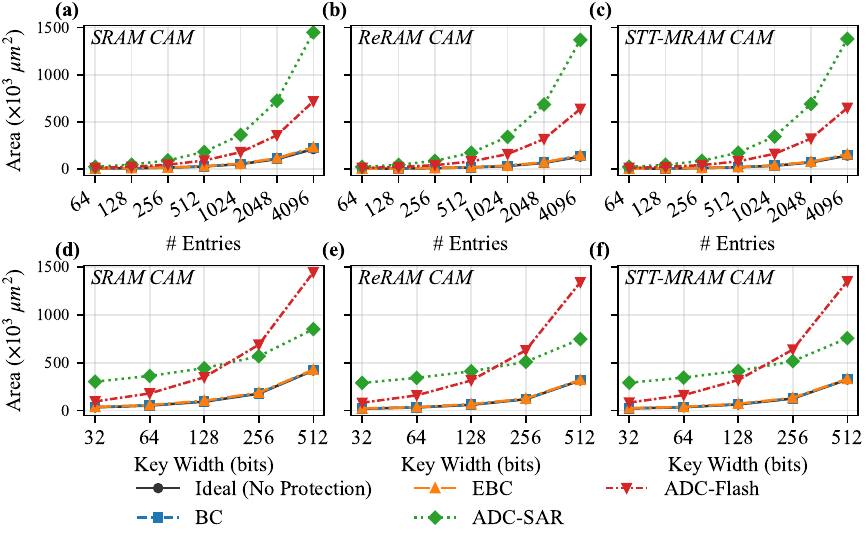}
  \vspace{-0.3in}
  \caption{
    (a)--(c)~Area overhead across various entry counts with 68-bit width;
    (d)--(f)~Area overhead across various entry widths with 1K entries;}
  \label{fig: area_combined}
\end{figure}

\paragraph{\textbf{Hardware Overhead.}}
Table~\ref{tab:area_breakdown} presents the area breakdown of each protection scheme under three hardware scopes: \textit{CAM-only}, including only the CAM array; \textit{CAM+Mem}, including the CAM array and in-memory ECC parity storage; and \textit{Overall}, representing the complete design with additional auxiliary components, such as ADCs, in peripheral circuitry.
In particular, under the \textit{CAM-only} scope, our scheme manifests slightly higher overhead (e.g., 6\% for 1b BC vs. 1.5\% for 1b ADC-SAR) due to the use of balanced codes, which have marginally higher overhead than traditional ECC. When incorporating in-memory parity, schemes have the same overhead increase (i.e., 3.9\% for 1b schemes and 7.7\% for 2b schemes), as they all employ the same correction code.
However, once peripheral overhead is included (the \textit{Overall} scope), the trend reverses: ADC-based designs require an ADC for each CAM entry, causing the auxiliary hardware overhead alone to exceed 300\% in our evaluated setting, even higher than TMR.
Figure~\ref{fig: area_combined} further confirms this trend across different entry counts and widths. In particular, as CAM capacity scales up, the area advantage of our scheme becomes increasingly pronounced.

\section{Conclusion}

In conclusion, we introduce a balanced-code protection scheme, which leverages balanced encoding to enable bit-level error detection in CAM. The key designs, BC and EBC, achieve a constant number of searches without modifying the underlying array structure. This work establishes the feasibility of protection code design for CAM and opens a new direction for reliable associative memory systems.

\bibliographystyle{ACM-Reference-Format}
\bibliography{bibs/fan,bibs/xin}
\end{document}